\documentclass[review]{elsarticle}

\usepackage{lineno,hyperref}
\modulolinenumbers[5]
\usepackage{amsmath}
\usepackage{amssymb}
\usepackage{braket}
\usepackage{color}
\usepackage{float}
\usepackage{epstopdf}
\usepackage{amsthm}
\usepackage{bm}
\usepackage{graphics,epstopdf}
\usepackage{color}
\usepackage[usenames,dvipsnames,svgnames]{pstricks}
\usepackage{epsfig}
\usepackage{pst-grad} 
\usepackage{pst-plot}
\newtheorem{Corollary}{Corollary}
\newtheorem{thm}{Theorem}
\newtheorem{lemma}{Lemma}

\journal{Reports on Mathematical Physics}

\begin{document}

\begin{frontmatter}

\title{Construction of infinitely many models of the universe on the Riemann hypothesis}

\author{Namrata Shukla}
\address{Harish-chandra Research Institute,\\ 
Chhatnag Road, Jhunsi,  Allahabad, 211 019,  India.}


\ead{namratashukla@hri.res.in}

\begin{abstract}The aim of this note is to remove an implausible assumption in Moser's theorem \cite{JM} to establish our new theorem 1
which gives a lower estimate for the sum $p+c^2\rho$ on Riemann hypothesis. Corollary 1 gives a rather plausible construction 
of infinitely many models of universe with positive density($\rho$) and 
pressure($p$) since it makes use of the state equation in the form of an inequality.
\end{abstract}

\begin{keyword}
\texttt Riemann hypothesis\sep Riemann-Siegel function \sep State equation\sep Relativistic cosmology.
\MSC[2010] 11M26\sep 11Z05
\end{keyword}

\end{frontmatter}

\linenumbers

\section{Introduction}
Let  $\zeta(s)$ denote the Riemann zeta-function with  $s=\sigma+it$ a complex variable \cite{Ti}. It is known that $\zeta(s)$ has infinitely many zeros, denoted by $\mathcal{\bm \rho}=\beta+i\gamma$ in the critical strip $0<\sigma<1$ and a positive proportion of them lie on the critical line $\sigma=\frac{1}{2}$.
The Riemann hypothesis states that all the non-trivial zeros of $\zeta(s)$ lie on the critical line. We abbreviate in this paper, ``under the truth of the Riemann hypothesis'' to ``on the RH'', which means that all the non-trivial zeros are of the form $\mathcal{\bm \rho}=\frac{1}{2}+i\gamma$.\\

We introduce the standard notation on the RH:
\begin{equation}\label{eqn1moser}
\chi\left(\frac{1}{2}+it \right)=\pi^{it}\frac{\Gamma\left(\frac{1}{4}-i\frac{t}{2} \right)}{\Gamma\left(\frac{1}{4}+i\frac{t}{2} \right)},
\end{equation}
and 
\begin{equation}\label{eqn2moser}
\vartheta(t)=\frac{1}{2}\mathop{\rm \arg}\chi\left(\frac{1}{2}+it \right),
\end{equation}
where $\Gamma$ denotes the gamma function. Then, we can write
\begin{equation}\label{eqn3moser}
\zeta\left(\frac{1}{2}+it \right)=e^{-i\vartheta(t)}Z(t).
\end{equation}
Here $Z(t)$ is the so-called `Riemann-Siegel' function and it is known that $Z(t)$ is a real differentiable and hence a continuous function and so infinitesimal calculus is applicable to it. Let $0<\gamma'<\gamma''$ be ordinates of the two neighboring zeroes of the zeta function.
Then $Z(\gamma^\prime)=Z(\gamma'')=0$ and hence by {\it Rolle's theorem} it follows that there exits a  $t_0$ in the interval $(\gamma^\prime,\gamma'')$ such that  $Z^\prime(t_0)= 0$.

J. Moser \cite{JM} studied the behavior of the Riemann zeta function on the critical line and applied it to relativistic cosmology on the RH in the last portion of the paper 
in order to show the existence of infinitely many universes. While proving this fact Moser says that by the Littlewood-Titchmarsh $\Omega$- theorem \cite[p. 204]{Ti}, we have the following:
Let $\{ \tilde{t}_0\}$ be a subsequence of the sequence $\{ t_0\}$ such that
\begin{equation}
\left|\zeta\left(\frac{1}{2}+i{\tilde{t}}_0 \right)  \right|>\frac{1}{\tilde{t}_0^\alpha}, \quad  0<\alpha \le 1.
\end{equation}
or, it can also be expressed as 
\begin{equation}\label{eqn4moser}
|Z(\tilde{t}_0)|>A_2\exp\left(\log^\beta \tilde{t}_0  \right), \quad \beta<\frac{1}{2},
\end{equation}
where $A_2$ is a positive constant and subsequently we follow the notation of Moser to denote a positive constant with or without suffixes. 
However, from `Titchmarsh omega theorem' there is no guarantee that one can choose such a sequence $\{ \tilde{t}_0\}$ as a 
subsequence of the sequence $\{ t_0\}$ consisting of stationary points of $Z(s)$. 
Since his argument about construction of infinitely many universes ultimately depends on the existence of the sequence $\{ \tilde{t}_0 \}$ for which both
\begin{equation}\label{eqn5moser}
Z^\prime(\tilde{t}_0)=0
\end{equation}
and \eqref{eqn4moser} hold and there may not be such a sequence, it is not possible to confirm the truth of this result.

The aim of this note is to remove this unproven hypothesis and confirm the construction of infinitely many models of universe 
with postive density and pressure without a requirement of \eqref{eqn5moser} to be fulfilled.

\begin{thm}
There exists an infinite sequence $\{t_0\}$ such that

\begin{equation}\label{theorem}
(p+ c^2\rho)>A_1{(\log \log\log t_0)}^2,
\end{equation}
where $c$ is the speed of light in vacuum, $\rho$ and $p$ are the density of matter 
and the pressure as functions of time and are also denoted by $\rho(t)$ and $p(t)$, respectively. 
\end{thm}

\begin{Corollary}\label{cor1.1moser}
If $p$ and $\rho$ are related by $\mathcal C_1\rho<p\leq \frac{c^2}{3}\rho$, $\mathcal C_1>0$ being a constant, then there exists a sequence $\{\delta({t_0})\}$ with $\delta({t}_0)>0$ such that
\begin{equation}\label{eqn6moser}
p(t)>A_1{(\log \log\log t_0)}^2>0, \quad t \in [{t}_0-\delta({t}_0)<{t}_0+\delta({t}_0)],
\end{equation}
{\it i.e.}, there exists an infinitely many models of the universe with positive pressure and density.\\
\end{Corollary}

The basis of relativistic cosmology is the following system of A. A. Friedman equations
\begin{equation}\label{eqn7moser}
\mathcal G c^2\rho =\frac{3}{R^2}\left(kc^2{+R^\prime}^2  \right),
\end{equation}

\begin{equation}\label{eqn8moser}
\mathcal G p =-\frac{2R''}{R}-\frac{{R^\prime}^2}{R^2}-\frac{kc^2}{R^2}
\end{equation}
where the cosmological constant is not being considered. $\mathcal G$ represents the gravitational constant and the curvature index $k$ can take the values amongst $(-1,0,1)$. 
The scale factor $R(t)$ denotes the radius of curvature of the universe for $k=1$, $p(t)$ and $\rho(t)$ are defined above and the prime indicates the derivative w.r.t. time.
Therefore, there are two equations for three unknown variables of time  $\rho(t), p(t)$ and $R(t)) $. In order to find the solution, we may eliminate one of them on
postulating the ``state equation" given by
 \begin{equation}\label{eqn9moser}
F(\rho, p)=0.
 \end{equation}
We usually study a model of the universe with zero pressure but the behavior of the universe is also studied in the case, {\it e.g.}, \cite[p. 387]{LL}
\begin{equation}\label{eqn10moser}
p=\frac{c^2\rho}{3},
\end{equation}
{\it i.e.}, in the case of the maximum pressure $p(t)$ under a given density $\rho(t)$. This restriction is natural from physics point of view. 
We also need essential physical restrictions to the density and pressure given by
\begin{equation}\label{eqn11moser}
\rho>0,  \quad  p\ge 0.
\end{equation}
Since $\rho>0$ is true if only $R(t)\ne 0$, we postulate the behavior of the radius of the universe so that the physical condition on pressure given by 
the second inequality in \eqref{eqn11moser} is satisfied in a certain interval of time. We construct (on the RH) an infinite set of models of 
the universe under the postulate on $R(t)$:

\begin{equation}\label{eqn2.16moser}
R(t)=|Z(t)|,  \quad  k=+1,
\end{equation}
and viewing the time variable $t$ as the imaginary part $t$ of the complex variable $s$. 

\section{Proof of the result}

\begin{lemma}\label{lemma2.1}
We have \cite[p. 308]{JM5} for $t\rightarrow\infty$
\begin{align}\label{eqn4JM}
\frac{\zeta^{''}\left(\frac{1}{2}+it \right)}{\zeta\left(\frac{1}{2}+it \right)}
=\sum_{\gamma}^{}\frac{1}{{(t-\gamma)}^2}+{\left\{\frac{ \zeta^\prime\left(\frac{1}{2}+it \right)}{\zeta\left(\frac{1}{2}+it \right)}\right\}}^2+O\left(\frac{1}{t} \right),
\end{align}
where the sum is over all ordinates of the non-trivial zeros of $\zeta(s)$ and $t$ does not coincide with any of them. This condition is valid throughout.
\end{lemma}

\begin{proof}
We apply the formula for {\it differentiation of the quotient} in the form
\begin{equation}\label{eqnzeta2.13}
{\left(\frac{f}{g}\right)}^\prime=\frac{f^\prime}{g}-\frac{fg^\prime}{g^2}.
\end{equation}
Riemann's {\it explicit formula} which is valid for all $s\ne \rho$ is
\begin{equation}\label{eqnzeta2.15}
\frac{\zeta^\prime}{\zeta}(s)=b-\frac{1}{s-1}-\frac{\rm d}{{\rm d}s}\log\Gamma\left(\frac{s}{2}+1 \right)
+\sum_{\rho}^{}\left(\frac{1}{s-\rho}+\frac{1}{\rho} \right),
\end{equation}
where $\rho$ runs through all non-trivial zeros of the Riemann zeta function with
\begin{equation}
\frac{\rm d}{{\rm d}s}\log\Gamma\left(s \right)=\log s-\frac{1}{2s}-2\int_{0}^{\infty}\frac{u}{(u^2+s^2)(e^{2\pi u}-1)}\, {\rm d}u,
\end{equation}
It is explicit because the logarithmic derivative of the Riemann zeta function is expressed explicitly 
in terms of the non-trivial zeros, while trivial zeros are contained in the second formula. Differentiating both sides of the first equality in \eqref{eqnzeta2.15} and using \eqref{eqnzeta2.13}, we deduce that
\begin{equation}\label{eqnzeta2.16}
\frac{\zeta''}{\zeta}(s)=\sum_{\rho}^{}\frac{1}{{(s-\rho)}^2}+{\Big\{\frac{\zeta^\prime}{\zeta}(s)\Big\}}^2+\frac{1}{{(s-1)}^2}-\frac{\rm d^2}{{\rm d}^2s}\log\Gamma\left(\frac{s}{2}+1 \right).
\end{equation}
The third term on the right can be estimated as $O\left(\frac{1}{|s|} \right)$. Hence, in particular, on the RH, we deduce \eqref{eqn4JM}.\\
\end{proof}

We may now give a simplified proof of Moser's main formula.
\begin{lemma}\label{lemma2.2}
On the RH we have \cite[p. 34]{JM}
\begin{equation}\label{eqn1JM}
\sum_{\gamma}^{}\frac{1}{{(t-\gamma)}^2}=-\frac{{\rm d}}{{\rm d}t}\left(\frac{Z^\prime(t)}{Z(t)} \right)+O\left(\frac{1}{t} \right).
\end{equation}
\end{lemma}
\begin{proof}
The logarithmic derivative of 
\begin{equation}
\zeta\left(\frac{1}{2}+it \right)=e^{-i\vartheta(t)}Z(t)
\end{equation}
gives
\begin{equation}\label{eqnmozer2'-1}
\frac{\zeta^\prime\left(\frac{1}{2}+it \right)}{\zeta\left(\frac{1}{2}+it \right)}=\vartheta^\prime(t)-\frac{Z^\prime(t)}{Z(t)}.
\end{equation}
We now apply formula \eqref{eqnzeta2.13} to deduce
\begin{align}\label{eqnmozer2'-2}
&\frac{{\rm d}}{{\rm d}t}\left(\frac{Z^\prime}{Z} \right)=\frac{Z^{''}(t)}{Z(t)}-{\left\{ \frac{Z^\prime(t)}{Z(t)}  \right\}}^2
=i\vartheta^{''}(t)-\frac{{\rm d}}{{\rm d}t}\left(\frac{\zeta^\prime}{\zeta} \right) \\ \nonumber
&=i\vartheta^{''}(t)-\frac{\zeta^{''}\left(\frac{1}{2}+it \right)}{\zeta\left(\frac{1}{2}+it \right)}
+{\left\{\frac{ \zeta^\prime\left(\frac{1}{2}+it \right)}{\zeta\left(\frac{1}{2}+it \right)}\right\}}^2.
\end{align}
We now deduce \eqref{eqn1JM} on comparing \eqref{eqnmozer2'-2} and \eqref{eqn4JM}. 
\end{proof}

\begin{lemma}\textup{(Littlewood's estimate \cite[p. 223]{Ti})}\label{lemma2.3}
 For ordinates $\gamma^\prime<\gamma''$ of consecutive zeros of the Riemann zeta function we have

\begin{equation}\label{eqnlittlewood}
\gamma''-\gamma^\prime<\frac{A}{\log \log \log \gamma^\prime}.
\end{equation}
\end{lemma}

\begin{lemma}\textup{(Titchmarsh $\Omega$- theorem \cite[p. 201]{Ti})}\label{lemma2.4}
There exists a sequence $\{{t}_0\}$  such that
\begin{equation}\label{rightomegatheorem}
|Z({t}_0)|>A_2\exp\left(\log^\beta {t}_0  \right), \quad \beta<\frac{1}{2}.
\end{equation}
\end{lemma}

The basis of our construction of infinitely many universes is the postulate \eqref{eqn2.16moser}
{\it i.e.} the radius of the universe coincides with the absolute value of the Riemann zeta-function $\left|\zeta\left(\frac{1}{2}+it  \right)\right|$, 
with spherical geometry being assumed.\\

Note that logarithmic differentiation applied to \eqref{eqn2.16moser} gives
$\log R(t)=\log |Z(t)|$ which implies $\frac{R^{\prime }(t)}{R(t)}=\frac{Z^{\prime }(t)}{Z(t)}$ and hence, \eqref{eqn8moser} takes the form
\begin{equation}\label{eqn7'moser}
\frac{\mathcal G c^2}{3}\rho(t)=\frac{c^2}{Z^2(t)}+{\left\{\frac{Z^\prime(t)}{Z(t)} \right\}}^2.
\end{equation}
We also have the simple relation between \eqref{eqn7moser} and \eqref{eqn8moser}
\begin{equation}\label{eqn21}
\mathcal G p =-2\frac{R''}{R}-\mathcal G \frac{c^2}{3}\rho.
\end{equation}
By Lemma \ref{lemma2.2} for $t=t_0$, we obtain
\begin{equation}
\mathcal G p +\mathcal G \frac{c^2}{3}\rho=-2\frac{R''}{R}=2\sum_{\gamma}^{}\frac{1}{{(t_0-\gamma)}^2}-2{\left\{\frac{Z^\prime(t_0)}{Z(t_0)} \right\}}^2+O\left(\frac{1}{t_0} \right)
\end{equation}
or
\begin{equation}\label{eqn22}
\frac{\mathcal G}{2} (p+ c^2\rho)=2\sum_{\gamma}^{}\frac{1}{{(t_0-\gamma)}^2}-\frac{1}{c^2}\frac{1}{{Z(t_0)}^2}+O\left(\frac{1}{t_0} \right).\\
\end{equation} 

Now we need to find an interval of time $t$ for which the physical conditions \eqref{eqn11moser} are satisfied. By Lemma \ref{lemma2.3} we obtain
\begin{equation}\label{eqn22moser}
\sum_{\gamma}^{}\frac{1}{{(t_0-\gamma)}^2}>\frac{1}{{(\gamma^{''}-\gamma^\prime)}^2}>A_1{(\log \log\log t_0)}^2.
\end{equation}
From \eqref{eqn22moser} and Lemma \ref{lemma2.4} we deduce that in a small neighborhood $[{t}_0-\delta({t}_0),~{t}_0+\delta({t}_0)]$ of the point ${t}_0$ we have
 \begin{equation}\label{eqnzeta2.17}
\frac{\mathcal G}{2} (p+ c^2\rho)>A_1{(\log \log\log t_0)}^2-A_2^\prime\exp\left(-\log^\beta {t}_0  \right)+O\left(\frac{1}{t_0} \right).
 \end{equation}
Since for large values of ${t}_0$, the first term on the right supersedes the other two terms, we conclude Theorem 1. By the inequality between $p$ and $\rho$, we conclude that $p(t)>0$ for infinitely 
many small intervals $[t_0-\delta(t_0),t_0+\delta(t_0)]$, whence Corollary 1 follows.



\section{Acknowledgement}
The author acknowledges Shigeru Kanemitsu for the stimulating ideas and support and also Faculty of Humanity-oriented Science and Technology, Kindai University, Iizuka, Fukuoka to sponsor her the visit in order to execute this work. 
She also thanks Kalyan Chakraborty from HRI, India for fruitful discussions.

\end{document}